\documentclass[11pt]{article}
\usepackage{amsmath,amssymb,amsthm,fullpage}
\newcommand{\namedref}[2]{\hyperref[#2]{#1~\ref*{#2}}}

\newcommand{\Supp}{{\sf Supp}}
\newcommand{\Span}{{\sf Span}}
\newcommand{\F}{\mathbb{F}}

\newcommand{\wt}{{\sf wt}}

\DeclareMathOperator{\Rk}{Rank}

\DeclareMathOperator{\loc}{Loc}

\theoremstyle{plain}

\newtheorem{theorem}{Theorem}
\newtheorem{lemma}[theorem]{Lemma}

\newtheorem{corollary}[theorem]{Corollary}
\newtheorem{definition}[theorem]{Definition}

\newtheorem{fact}[theorem]{Fact}

\theoremstyle{definition}
\theoremstyle{plain}

\newcommand{\mb}[1]{\ensuremath{\mathbf{#1}}}
\newcommand{\mc}[1]{\ensuremath{\mathcal{#1}}}

\newcommand{\eat}[1]{}

\renewcommand{\geq}{\geqslant}
\renewcommand{\leq}{\leqslant}

\renewcommand{\epsilon}{\varepsilon}

\def\bull{\vrule height .9ex width .8ex depth -.1ex }

\title{On the Locality of Codeword Symbols}
\author{Parikshit Gopalan \\ Microsoft Research \\ parik@microsoft.com
\and Cheng Huang \\ Microsoft Research \\ chengh@microsoft.com
\and Huseyin Simitci \\ Microsoft Corporation \\ huseyins@microsoft.com
\and Sergey Yekhanin \\ Microsoft Research \\ yekhanin@microsoft.com}

\begin{document}
\date{}
\maketitle

\begin{abstract}
Consider a linear $[n,k,d]_q$ code $\mc{C}.$ We say that that $i$-th
coordinate of $\mc{C}$ has locality $r,$ if the value at this
coordinate can be recovered from accessing some other $r$ coordinates
of $\mc{C}.$ Data storage applications require codes with small
redundancy, low locality for information coordinates, large distance,
and low locality for parity coordinates. In this paper we carry out an
in-depth study of the relations between these parameters.

We establish a tight bound for the redundancy $n-k$ in terms of the
message length, the distance, and the locality of information
coordinates. We refer to codes attaining the bound as optimal. We
prove some structure theorems about optimal codes, which are
particularly strong for small distances. This gives a fairly complete
picture of the tradeoffs between codewords length, worst-case distance
and locality of information symbols.

We then consider the locality of parity check symbols and erasure
correction beyond worst case distance for optimal codes.
Using our structure theorem, we obtain a tight bound for the
locality of parity symbols possible in such codes for a broad class of
parameter settings. We prove that there is a tradeoff between having good
locality for parity checks and the ability to correct erasures beyond
the minimum distance.

\end{abstract}

\section{Introduction}\label{Sec:Introduction}

Modern large scale distributed storage systems such as data centers store data in a redundant form to ensure reliability against node (e.g., individual machine) failures. The simplest solution here is the straightforward replication of data packets across different nodes. Alternative solution involves erasure coding: the data is partitioned into $k$ information packets. Subsequently, using an erasure code, $n-k$ parity packets are generated and all $n$ packets are stored in different nodes.

Using erasures codes instead of replication may lead to dramatic improvements both in terms of redundancy and reliability. However to realize these improvements one has to address the challenge of maintaining an erasure encoded representation. In particular, when a node storing some packet fails, one has to be able to quickly reconstruct the lost packet in order to keep the data readily available for the users and to maintain the same level of redundancy in the system. We say that a certain packet has {\it locality} $r$ if it can be recovered from accessing only $r$ other packets. One way to ensure fast reconstruction is to use erasure codes where all packets have low locality $r\ll k.$ Having small value of locality is particularly important for information packets.

These considerations lead us to introduce the concept of an
$(r,d)$-{\it code}, i.e., a linear code of distance~$d,$ where all
information symbols have locality at most $r.$ Storage system based on
$(r,d)$-codes provide fast recovery of information packets from a
single node failure (typical scenario), and ensure that no data is
lost even if up to $d-1$ nodes fail simultaneously. One specific class
of $(r,d)$-codes called Pyramid Codes has been considered in~\cite{HCL}.

Pyramid codes can be obtained from any systematic Maxmimum Distance Seperable (MDS) codes of distance $d,$ such as Reed Solomon codes. Assume for simplicity that the first parity check symbol is the sum $\sum_{i=1}^kx_i$ of the information symbols. Replace this with $\left\lceil\frac{k}{r}\right\rceil$ parity checks each of size at most $r$ on disjoint information symbols. It is not hard to see that the resulting code $\mc{C}$ has information locality $r$ and distance $d$, while  the redundancy of the code $\mc{C}$ is given by
\begin{equation}
\label{Eqn:PyramidRedundancy}
n-k = \left\lceil \frac{k}{r} \right\rceil +  d-2.
\end{equation}

\subsection{Our results}

In this paper we carry out an in-depth study of the relations between redundancy, erasure-correction and symbol locality in linear codes.

Our first result is a tight bound for the redundancy in terms of the
message length, the distance, and the information locality. We show
that in any $[n,k,d]_q$ code of information locality $r,$
\begin{equation}
\label{Eqn:IntoMain}
n-k \geq \left\lceil \frac{k}{r} \right\rceil +  d-2.
\end{equation}
We refer to codes attaining the bound above as optimal. Pyramid codes
are one such family of codes. The bound~(\ref{Eqn:IntoMain}) is of
particular interest in the case when $r\mid k,$ since otherwise one
can improve the code by increasing the dimension while maintaining the
$(r,d)$-property and redundancy intact. A closer examination of our
lower bound gives a structure theorem for optimal codes when
$r\mid k$. This theorem is especially strong when $d < r+3$, it fixes
the support of the parity check matrix, the only freedom is in the
choice of coefficients. We also show that the condition $r < d+3$ is
in fact necessary for such a strong statement to hold.

We then turn our attention to the locality of parity symbols.
We prove tight bounds on the locality of parity symbols in optimal codes assuming
$d < r+3.$ In particular we establish the existence of optimal $(r,d)$-codes
that are significantly better than pyramid codes with respect to
locality of parity symbols. Our codes are explicit in the case of
$d=4,$ and non-explicit otherwise. The lower bound is proved using the
structure theorem. Finally, we relax the conditions
$d<r+3$ and $r\mid k$ and exhibit one specific family of optimal codes
that gives locality $r$ for all symbols.

Our last result concerns erasure correction beyond the worst case distance of the code. Assume that we are given a bipartite graph which describes the supports of the parity check symbols. What choice of coefficients will maximize the set of erasure patterns that can be corrected by such a code? In~\cite{HCL} the authors gave a necessary condition for an erasure pattern to be correctable, and showed that over sufficiently large fields, this condition is also sufficient. They called such codes Generalized Pyramid codes. We show that such codes cannot have any non-trivial parity locality; thus establishing a tradeoff between parity locality and erasure correction beyond the worst case distance.

\subsection{Related work}

There are two classes of erasure codes providing fast recovery procedures for individual codeword coordinates (packets) in the literature.

{\it Regenerating codes.} These codes were introduced in~\cite{Dimakis_1} and developed further in e.g.,~\cite{RSK,CJM}. See~\cite{Dimakis_survey} for a survey. One crucial idea behind regenerating codes is that of {\it sub-packetization}. Each packet is composed of few sub-packets, and when a node storing a packet fails all (or most of) other nodes send in some of their sub-packets for recovery. Efficiency of the recovery procedure is measured in terms of the overall bandwidth consumption, i.e., the total size of sub-packets required to recover from a single failure. Somehow surprisingly regenerating codes can in many cases achieve a rather significant reduction in bandwidth, compared with codes that do not employ sub-packetization. Our experience with data centers however suggests that in practice there is a considerable overhead related to accessing extra storage nodes. Therefore
pure bandwidth consumption is not necessarily the right single measure of the recovery time. In particular, coding solutions that do not rely on sub-packetization and thus access less nodes (but download more data) are sometimes more attractive.

\smallskip

{\it Locally decodable codes.} These codes were introduced in~\cite{KT} and developed further in e.g.,~\cite{Y_nice,Efremenko,KSY}. See~\cite{Y_now} for a survey. An $r$-query Locally Decodable Code (LDC) encodes messages in such a way that one can recover any message symbol by accessing only $r$ codeword symbols even after some arbitrarily chosen (say) $10\%$ of codeword coordinates are erased. Thus LDCs are in fact very similar to $(r,d)$-codes addressed in the current paper, with an important distinction that LDCs allow for local recovery even after a very large number of symbols is erased, while $(r,d)$-codes provide locality only after a single erasure. Not surprisingly locally decodable codes require substantially larger codeword lengths then $(r,d)$-codes.

\subsection{Organization}

In section~\ref{Sec:LBoundAndStructure} we establish the lower bound for redundancy of $(r,d)$-codes and obtain a structural characterization of optimal codes, i.e., codes attaining the bound. In section~\ref{Sec:CanonicalCodes} we strengthen the structural characterization for optimal codes with $d<r+3$ and show that any such code has to be a canonical code. In section~\ref{Sec:ParityLocality} we prove matching lower and upper bounds on the locality of parity symbols in canonical codes. Our code construction is not explicit and requires the underlying field to be fairly large. In the special case of codes of distance $d=4,$ we come up with an explicit family that does not need a large field. In section~\ref{Sec:NonCanonicalCodes} we present one optimal family of non-canonical codes that gives uniform locality for all codeword symbols. Finally, in section~\ref{Sec:BeyondDist} we study erasure correction beyond the worst case distance and prove that systematic codes correcting the maximal number of erasure patterns (conditioned on the support structure of the generator matrix) cannot have any non-trivial locality for parity symbols.

\section{Preliminaries}

We use standard mathematical notation
\begin{itemize}

\item For an integer $t,$ $[t]=\{1,\ldots,t\};$

\item For a vector ${\bf x},$ $\Supp({\bf x})$ denotes the set $\{i: {\bf x}_i \neq 0\};$

\item For a vector ${\bf x},$ $\wt({\bf x}) = |\Supp({\bf x})|$ denotes the Hamming weight;

\item For a vector ${\bf x}$ and an integer $i,$ ${\bf x}(i)$ denotes the $i$-th coordinate of ${\bf x};$

\item For sets $A$ and $B,$ $A\sqcup B$ denotes the disjoint union.
\end{itemize}

Let $\mc{C}$ be an $[n,k,d]_q$ linear code. Assume that the encoding of ${\bf x}\in \mathbb{F}_q^k$ is by the vector
\begin{equation}
\label{Eqn:CodeDefinition}
\mc{C}({\bf x}) = ({\bf c}_1\cdot {\bf x}, {\bf c}_2\cdot {\bf x},\ldots, {\bf c}_n\cdot {\bf x}) \in \F_q^n.
\end{equation}
Thus the code $\mc{C}$ is specified by the set of $n$ points $C = \{{\bf c}_1,\ldots,{\bf c}_n\} \in \F_q^k$. The set of points must have full rank for $\mc{C}$ to have $k$ information symbols. It is well known that the distance property is captured by the following condition
(e.g.,~\cite[theorem 1.1.6]{TVN_book}).

\begin{fact}
\label{Fact:Dist}
The code $\mc{C}$ has distance $d$ if and only if for every $S \subseteq C$ such that $\Rk(S) \leq k-1,$
\begin{equation}
\label{Eqn:SnD}
|S| \leq n-d.
\end{equation}
\end{fact}

In other words, every hyperplane through the origin misses at least $d$ points from $C$. In this work, we are interested in the recovery cost of each symbol in the code from a single erasure.

\begin{definition}\label{Def:Locality}
For ${\bf c}_i \in C$, we define $\loc({\bf c}_i)$ to be the smallest integer $r$ for which there exists $R \subseteq C$ of cardinality $r$ such that
$${\bf c}_i = \sum_{j \in R} \lambda_j{\bf c}_j.$$
We further define $\loc(\mc{C}) = \max_{i \in [n]}\loc({\bf c}_i)$.
\end{definition}

Note that $\loc({\bf c}_i) \leq k$, provided $d \geq 2$, since this guarantees that $C \setminus \{{\bf c}_i\}$ has full dimension. We will be interested in (systematic) codes which guarantee locality for the information symbols.

\begin{definition}
We say that a code $\mc{C}$ has information locality $r$ if there exists $I \subseteq C$ of full rank such that $\loc({\bf c}) \leq r$ for all ${\bf c} \in I.$
\end{definition}

For such a code we can choose $I$ as our basis for $\F_q^k$ and partition $C$ into $I = \{{\bf e}_1,\ldots,{\bf e}_k\}$ corresponding to information symbols and $C\setminus I = \{{\bf c}_{k+1},\ldots,{\bf c}_n\}$ corresponding to parity check symbols. Thus the code $\mc{C}$ can be made systematic. \begin{definition}
\label{Def:RDCode}
A code $\mc{C}$ is an $(r,d)$-code if it has information locality $r$ and distance $d.$
\end{definition}

For any code $\mc{C}$, the set of all linear dependencies of length at most $r+1$ on points in $C$ defines a natural hypergraph $H_r(V,E)$ whose vertex set $V = [n]$ is in one-to-one correspondence to points in $C$. There is an edge corresponding to set $S \subseteq V$ if $|S| \leq r +1$
$$\sum_{i \in S}\lambda_i{\bf c}_i = 0, \ \ \lambda_i \neq 0.$$
Equivalently $S \subseteq [n]$ is an edge in $H$ if it supports a codeword in $\mc{C}^\perp$ of weight at most $r+1$. Since $r$ will usually be clear from the context, we will just say $H(V,E)$. A code $\mc{C}$ has locality $r$ if there are no isolated vertices in $H$. A code $\mc{C}$ has information locality $r$ if the set points corresponding to vertices that are incident to some edge in $H$ has full rank.

We conclude this section presenting one specific class of $(r,d)$-codes has been considered in~\cite{HCL}:

{\it Pyramid codes.} In what follows the dot product of vectors ${\bf p}$ and ${\bf x}$ is denoted by ${\bf p}\cdot {\bf x}.$ To define an $(r,d)$ pyramid code $\mc{C}$ encoding messages of dimension $k$ we fix an arbitrary linear systematic $[k+d-1,k,d]_q$ code $\mc{E}.$ Clearly, $\mc{E}$ is MDS. Let
$$\mc{E}({\bf x}) = ({\bf x},{\bf p}_0\cdot {\bf x},{\bf p}_1\cdot {\bf x},\ldots,{\bf p}_{d-2}\cdot {\bf x}).$$
We partition the set $[k]$ into $t=\left\lceil\frac{k}{r}\right\rceil$ subsets of size up to $r,$ $[k]=\bigsqcup_{i\in [t]} S_i.$ For a $k$-dimensional vector ${\bf x}$ and a set $S\subseteq [k]$ let ${\bf x}|_{S}$ denote the $|S|$-dimensional restriction of ${\bf x}$ to coordinates in the set $S.$ We define the systematic code $\mc{C}$ by
$$
\mc{C}({\bf x}) = \left({\bf x},\left({\bf p}_0|_{S_1}\cdot {\bf x}|_{S_1}\right),\ldots,\left({\bf p}_0|_{S_t}\cdot {\bf x}|_{S_t}\right),{\bf p}_1\cdot {\bf x},\ldots,{\bf p}_{d-2}\cdot {\bf x}\right).
$$
It is not hard to see that the code $\mc{C}$ has distance $d.$ To see
that all information symbols and the first
$\left\lceil\frac{k}{r}\right\rceil$ parity symbols of $\mc{C}$ have
locality $r$ one needs to observe that (since $\mc{E}$ is an MDS code)
the vector ${\bf p}_0$ has full Hamming weight. The last $d-2$ parity
symbols of $\mc{C}$ may have locality as large as $k.$

\section{Lower Bound and the Structure Theorem}\label{Sec:LBoundAndStructure}

We are interested in systematic codes with information locality $r$. Given $k,r,d$ our goal is to minimize the codeword length $n$. Since the code is systematic, this amounts to minimizing the redundancy $h = n -k$. Pyramid codes have $h = \left\lceil\frac{k}{r}\right\rceil + d -2.$
Our goal is to prove a matching lower bound. Lower bounds of $k/r$ and $d-1$ are easy to show, just from the locality and distance constraints respectively. The hard part is to sum them up.
\begin{theorem}
\label{Theorem:LowerBound}
For any $[n,k,d]_q$ linear code with information locality $r$,
\begin{equation}
\label{Eqn:LowerBound}
n-k \geq \left\lceil \frac{k}{r} \right\rceil +  d-2.
\end{equation}
\end{theorem}
\begin{proof}
Our lower bound proceeds by constructing a large set $S \subseteq C$ where $\Rk(S) \leq k-1$ and then applying Fact~\ref{Fact:Dist}. The set $S$ is constructed by the following algorithm:
\begin{center}
\fbox{%
\parbox{15cm}{%
1.\ \ Let $i =1, S_0 = \{\}$.

2.\ \ While $\Rk(S_{i-1}) \leq k-2$:

3.\ \ \hspace{1cm} Pick ${\bf c}_i \in C \setminus S_{i-1}$ such that there is a hyperedge $T_i$ in $H$ containing ${\bf c}_i.$

4.\ \ \hspace{1cm} If $\Rk(S_{i-1} \cup T_i) < k$, set $S_i = S_{i-1} \cup T_i$.

5.\ \ \hspace{1.4cm} Else pick $T^\prime \subset T_i$ so that $\Rk(S_{i-1} \cup T^\prime) = k-1$ and set $S_i = S_{i-1} \cup T^\prime$.

6.\ \ \hspace{1cm} Increment $i.$
}%
}
\end{center}
In Line 3, since $\Rk(S_{i-1}) \leq k-2$ and $\Rk(I) = k,$ there exists ${\bf c}_i$ as desired. Let $\ell$ denote the number of times the set $S_i$ is grown. Observe that the final set $S_\ell$ has $\Rk(S_\ell) = k-1$. We now lower bound $|S|$. We define $s_i, t_i$ to measure the increase in the
size and rank of $S_i$ respectively:
\begin{align*}
s_i = |S_i| - |S_{i-1}|, &  \ \ |S_\ell| = \sum_{i=1}^\ell s_i, \\
t_i = \Rk(S_i) - \Rk(S_{i-1}), & \  \ \Rk(S_\ell) = \sum_{i=1}^\ell t_i = k-1.
\end{align*}
We analyze two cases, depending on whether the condition $\Rk(S_{i-1}\cup T_i) = k$ is ever reached. Observe that this condition can only be reached when $i = \ell.$

{\bf Case 1:} Assume $\Rk(S_{i-1} \cup T_i) \leq k-1$ throughout. In each step we add $s_i \leq r+1$ vectors. Note that these vectors are always such that some nontrivial linear combination of them yields a (possibly zero) vector in $\Span(S_{i-1}).$ Therefore we have $t_i \leq s_i-1 \leq r.$ So there are $\ell \geq \left\lceil\frac{k-1}{r}\right\rceil$ steps in all. Thus
\begin{equation}
\label{Eqn:Case1}
|S| =\sum_{i=1}^\ell s_i \geq \sum_{i=1}^\ell (t_i  +1) \geq k-1 + \left\lceil \frac{k-1}{r}
\right\rceil
\end{equation}
Note that $k-1 + \left\lceil \frac{k-1}{r} \right\rceil \geq k + \left\lceil \frac{k}{r} \right\rceil -2$ with equality holding whenever $r=1$ or $k \equiv 1 \bmod r.$

{\bf Case 2:} In the last step, we hit the condition $\Rk(S_{\ell -1} \cup T_\ell) = k$. Since the rank only increases by $r$ per step, $\ell \geq \left\lceil \frac{k}{r} \right\rceil$. For $i \leq \ell -1$, we add a set $T_i$ of $s_i \leq r+1$ vectors. Again note that these vectors are always such that some nontrivial linear combination of them yields a (possibly zero) vector in $\Span(S_{i-1}).$ Therefore $\Rk(S_i)$ grows by $t_i$ where  $t_i \leq s_i -1.$ In Step $\ell$, we add $T^\prime \subset T_\ell$ to $S.$ This increases $\Rk(S)$ by $t_\ell \geq 1$ (since $\Rk(S) \leq k-2$ at the start) and $|S|$ by $s_\ell \geq t_\ell$. Thus
\begin{equation}
\label{Eqn:Case2}
|S| = \sum_{i=1}^\ell s_i \geq \sum_{i=1}^{\ell-1} (t_i +1) + t_\ell = k + \left\lceil \frac{k}{r} \right\rceil -2.
\end{equation}
The conclusion now follows from Fact~\ref{Fact:Dist} which implies that  $|S| \leq n-d$.
\end{proof}

\begin{definition}
We say that an $(r,d)$-code $\mc{C}$ is optimal if its parameters satisfy~(\ref{Eqn:LowerBound}) with equality.
\end{definition}

Pyramid codes~\cite{HCL} yield optimal $(r,d)$-codes for all values of $r,d,$ and $k$ when the alphabet $q$ is sufficiently large.

The proof of theorem~\ref{Theorem:LowerBound} reveals information about the structure of optimal $(r,d)$-codes. We think of the algorithm as attempting to maximize $$\frac{|S|}{\Rk(S)} = \frac{\sum_{i=1}^\ell s_i}{\sum_{i=1}^\ell t_i}.$$
With this in mind, at step $i$ we can choose ${\bf c}_i$ such that $\frac{s_i}{t_i}$ is maximized. An optimal length code should yield the same value for $|S|$ for this (or any) choice of ${\bf c}_i.$ This observation yields an insight into the structure of local dependencies in optimal codes, as given by the following structure theorem.

\begin{theorem}
\label{Theorem:Structure}
Let $\mc{C}$ be an $[n,k,d]_q$ code with information locality $r.$ Suppose $r\mid k,$ $r<k,$ and
\begin{equation}
\label{Eqn:StructureN}
n = k + \frac{k}{r} + d-2;
\end{equation}
then hyperedges in the hypergraph $H(V,E)$ are disjoint and each has size exactly $r+1.$
\end{theorem}
\begin{proof}
We execute the algorithm presented in the proof of theorem~\ref{Theorem:LowerBound} to obtain a set $S$ and sequences $\{s_i\}$ and $\{t_i\}.$ We consider the case of $r=1$ separately.  Since all $t_i\leq 1$ we fall into Case~1. Combining formulas~(\ref{Eqn:Case1}),~(\ref{Eqn:SnD}) and~(\ref{Eqn:StructureN}) we get
$$|S| =\sum_{i=1}^\ell s_i = \sum_{i=1}^\ell t_i  + \ell = 2k-2.$$
Combining this with $\sum_{i=1}^\ell t_i=k-1$ we conclude that $\ell=k-1,$ all $s_i$ equal $2,$ and all $t_i$ equal $1.$ The latter two conditions preclude the existence of hyperedges of size $1$ or intersecting edges in $H.$

We now proceed to the case of $r>1.$ When $r\mid k,$ the bound in equation~(\ref{Eqn:Case1}) is larger than that in equation~(\ref{Eqn:Case2}). Thus, we must be in Case~2. Combining formulas~(\ref{Eqn:Case2}),~(\ref{Eqn:SnD}) and~(\ref{Eqn:StructureN}) we get
$$|S| = \sum_{i=1}^\ell s_i = \sum_{i=1}^{\ell} t_i + \ell - 1 = k + \frac{k}{r} -2.$$
Observe that $\sum_{i=1}^\ell t_i=k-1$ and thus $\ell=\frac{k}{r}.$ Together with the constraint $t_i \leq r,$ this implies that $t_j =r-1$ for some $j \in [\ell]$ and $t_i =r$ for $i \neq j.$ We claim that in fact $j = \ell.$ Indeed, if $j < \ell,$ we would have $\sum_{i \leq \ell-1}t_i = k -r -1$ and
$t_\ell =r,$ hence we would be in Case~1.

Now assume that there is an edge $T$ with $|T| \leq r$. By adding this edge to $S$ at the first step, we would get $t_1 \leq r-1.$ Next assume that $T_1 \cap T_2$ is non-empty. Observe that this implies $\Rk(T_1 \cup T_2) < 2r.$ So if we add edges $T_1$ and $T_2$ to $S,$ we have $t_1 + t_2 \leq 2r-1.$ Clearly these conditions lead to contradiction if $\ell=\frac{k}{r} \geq 3.$ In fact, they also give a contradiction for $\frac{k}{r} = 2,$ since they put us in Case~1.
\end{proof}

\section{Canonical Codes}\label{Sec:CanonicalCodes}

The structure theorem implies then when $d$ is sufficiently small (which in our experience is the setting of interest in most data storage applications), optimal $(r,d)$-codes have rather rigid structure. We formalize this by defining the notion of a {\it canonical} code.

\begin{definition}
Let $\mc{C}$ be a systematic $[n,k,d]_q$ code with information locality $r$ where $r\mid k,$ $r<k,$ and $n = k+ \frac{k}{r} + d-2$. We say that $\mc{C}$ is canonical if the set $C$ can partitioned into three groups $C=I\cup C^\prime\cup C^{\prime\prime}$ such that:
\begin{enumerate}
\item Points $I = \{{\bf e}_1,\ldots,{\bf e}_k\}.$
\item Points $C^\prime = \{{\bf c}^\prime_1,\ldots,{\bf c}^\prime_{k/r}\}$ where $\wt({\bf c}^\prime_i) = r$. The supports of these vectors are disjoint sets which partition $[k].$
\item Points $C^{\prime\prime}= \{{\bf c}^{\prime\prime}_1,\ldots,{\bf c}^{\prime\prime}_{d-2}\}$ where $\wt({\bf c}^{\prime\prime}_i) = k.$
\end{enumerate}
\end{definition}

Clearly any canonical code is systematic and has information locality $r$. The distance property requires a suitable choice of vectors $\{{\bf c}^\prime\}$ and $\{{\bf c}^{\prime\prime}\}.$ Pyramid codes~\cite{HCL} are an example of canonical codes. We note that since $r<k,$ there is always a distinction between symbols $\{{\bf c}^\prime\}$ and $\{{\bf c}^{\prime\prime}\}.$

\begin{theorem}
\label{Theorem:Canonical}
Assume that $d < r+3,$ $r<k,$ and $r\mid k.$ Let  $n = k + \frac{k}{r} +  d-2$. Every systematic $[n,k,d]_q$ code with information locality $r$ is a canonical code.
\end{theorem}
\begin{proof}
Let $\mc{C}$ be a systematic $[n,k,d]$ code with information locality $r$. We start by showing that the hypergraph $H(V,E)$ has $\frac{k}{r}$ edges.

Since $\mc{C}$ is systematic, we know that $I = \{{\bf e}_1,\ldots,{\bf e}_k\} \subset C$. By theorem~\ref{Theorem:Structure}, $H(V,E)$ consists of
$m$ disjoint, $(r+1)$-regular edges and every vertex in $I$ appears in some edge. But since the points in $I$ are linearly independent, every edges involves at least one vertex from $C\setminus I$ and at most $r$ from $I$. So we have $m \geq \frac{k}{r}$. We show that equality holds.

Assume for contradiction that $m \geq \frac{k}{r} + 1$. Since the edges are regular and disjoint, we have
$$n \geq m(r+1) = k + \frac{k}{r} + r + 1 > k + \frac{k}{r} + d -2$$
which contradicts the choice of $n$. Thus $m = \frac{k}{r}.$ This means that every edge $T_i$ is incident on exactly $r$ vertices
${\bf e}_{i_1}\ldots,{\bf e}_{i_r}$ from $I$ and one vertex ${\bf c}^\prime_i$ outside
it. Hence
$${\bf c}^\prime_i  = \sum_{j=1}^r \lambda_{i_j}{\bf e}_{i_j}.$$
Since the $T_i$s are disjoint, the vectors ${\bf c}^\prime_1,\ldots,{\bf c}^\prime_{k/r}$ have disjoint supports which partition $[k]$.

We now show that the remaining vectors ${\bf c}^{\prime\prime}_1,\ldots,{\bf c}^{\prime\prime}_{d-2}$ must all have $\wt({\bf c}^{\prime\prime}) =k$. For this, we consider the encoding of ${\bf e}_j$. We note ${\bf e}_i\cdot {\bf e}_j \neq 0$ iff $i =j$ and ${\bf c}^\prime_i\cdot {\bf e}_j \neq 0$ iff $j \in \Supp({\bf c}_i)$. Thus only $2$ of these inner products
are non-zero. Since the code has distance $d$, all the $d-2$ inner products ${\bf c}^{\prime\prime}_i\cdot {\bf e}_j$ are non-zero. This shows that $\wt({\bf c}^{\prime\prime}_i) = k$ for
all $i$.
\end{proof}

The above bound is strong enough to separate having locality of $r$
from having just information locality of $r$. The following corollary
follows from the observation that the hypergraph $H(V,E)$ must contain
$n - \frac{k}{r}(r+1) = d-2$ isolated vertices, which do not
participate in any linear relations of size $r+1$.

\begin{corollary}
Assume that  $2< d < r+3$ and $r\mid k.$ Let $n = k + \frac{k}{r} +
d-2$. There are no $[n,k,d]_q$ linear codes with locality $r$.
\end{corollary}

\section{Canonical codes: parity locality}\label{Sec:ParityLocality}

Theorem~\ref{Theorem:Canonical} gives a very good understanding optimal $(r,d)$-codes in the case $r<d+3$ and $r\mid k.$ For any such code the coordinate set $C=\{{\bf c}_i\}_{i\in [n]}$ can be partitioned into sets $I,C^\prime,C^{\prime\prime}$ where for all ${\bf c}\in I\cup C^\prime,$ $\loc({\bf c})=r,$ and for all ${\bf c^{\prime\prime}}\in C^{\prime\prime},$ $\loc({\bf c}^{\prime\prime})>r.$ It is natural to ask how low can the locality of symbols ${\bf c}^{\prime\prime}\in C^{\prime\prime}$ be. In this section we address and resolve this question.

\subsection{Parity locality lower bound}

We begin with a lower bound.

\begin{theorem}
\label{Theorem:ParityLocalityLower}
Let $\mathcal C$ be a systematic optimal $(r,d)$-code with parameters
$[n,k,d]_q.$ Suppose $d<r+3,$ $r<k,$ and $r\mid k.$ Then some $\frac{k}{r}$ parity symbols of $\mc{C}$ have locality exactly $r,$ and $d-2$ other parity symbols of $\mc{C}$ have locality no less than
\begin{equation}
\label{Eqn:ParityLocLower}
k-\left(\frac{k}{r}-1\right)(d-3).
\end{equation}
\end{theorem}
\begin{proof}
Theorem~\ref{Theorem:Canonical} implies that $\mc{C}$ is a canonical code. Let $C=I\cup C^\prime \cup C^{\prime\prime}$ be the canonical partition of the coordinates of $\mc{C}.$ Clearly, for all $\frac{k}{r}$ symbols ${\bf c}^\prime \in C^\prime$ we have $\loc({\bf c}^\prime)\leq r.$ We now prove lower bounds on the locality of symbols in $C^\prime\cup C^{\prime\prime}.$

We start with symbols ${\bf c}^{\prime\prime}\in C^{\prime\prime}.$ For every $j\in [k/r]$ we define a subset $R_j\subseteq C$ that we call a {\it row}. Let $S_j=\Supp\left({\bf c}^{\prime}_j\right).$ The $j$-th row contains the vector ${\bf c}^\prime_j,$ all $r$ unit vectors in the support of ${\bf c}^\prime_j$ and the set $C^{\prime\prime}.$
$$
R_j = \{{\bf c}^\prime_j \}\cup \left(\bigcup_{i\in S_j} {\bf e}_i\right) \cup C^{\prime\prime}.
$$
Observe that restricted to $I\cup C^\prime$ rows $\{R_j\}_{j\in [k/r]}$ form a partition. Consider an arbitrary symbol ${\bf c}^{\prime\prime}\in C^{\prime\prime}.$ Let $\ell=\loc({\bf c}^{\prime\prime}).$ We have
\begin{equation}
\label{Eqn:CPrimeSum}
{\bf c}^{\prime\prime} = \sum_{i\in L} {\bf c}_i,
\end{equation}
where $|L|=\ell.$ In what follows we show that for each row $R_j,$
\begin{equation}
\label{Eqn:RowIntersect}
|R_j \cap L|\geq r
\end{equation}
needs to hold. It is not hard to see that this together with the structure of the sets $\{R_j\}$ implies inequality~(\ref{Eqn:ParityLocLower}). To prove~(\ref{Eqn:RowIntersect}) we consider the code
\begin{equation}
\label{Eqn:CodeRestrictJ}
\mc{C}_j=\{\mc{C}({\bf x})\mid {\bf x}\in \mathbb{F}_q^k\mathrm{\ such\ that\ }\Supp({\bf x})\subseteq S_j\}.
\end{equation}
It is not hard to see that $\Supp(\mc{C}_j)=R_j$ and $\dim\mc{C}_j=r.$ Observing that the distance of the code $\mc{C}_j$ is at least $d$ and $|R_j|=r+d-1$ we conclude that (restricted to its support) $\mc{C}_j$ is an MDS code. Thus any $r$ symbols of $\mc{C}_j$ are independent. It remains to note that~(\ref{Eqn:CPrimeSum}) restricted to coordinates in $S_j$ yields a non-trivial dependency of length at most $|R_j\cap L|+1$ between the symbols of $\mc{C}_j.$

We proceed to the lower bound on the locality of symbols in $C^\prime.$ Fix an arbitrary ${\bf c}^\prime_j \in C^\prime.$ A reasoning similar to the one above implies that if $\loc({\bf c}_j)<r;$ then there is a dependency of length below $r+1$ between the coordinates of the $[r+d-1,r,d]_q$ code $\mc{C}_j$ (defined by~(\ref{Eqn:CodeRestrictJ})) restricted to its support.
\end{proof}

Observe that the bound~(\ref{Eqn:ParityLocLower}) is close to $k$ only when $r$ is large and $d$ is small. In other cases theorem~\ref{Theorem:ParityLocalityLower} does not rule out existence of canonical codes with low locality for all symbols (including those in $C^{\prime\prime}$). In the next section we show that such codes indeed exist. In particular we show that the bound~(\ref{Eqn:ParityLocLower}) can be always met with equality.

\subsection{Parity locality upper bounds}

Our main results in this section are given by theorems~\ref{Theorem:ParityLocalityUpperGeneral} and~\ref{Theorem:ParityLocalityUpperD4}. Theorem~\ref{Theorem:ParityLocalityUpperGeneral} gives a general upper bound matching the lower bound of theorem~\ref{Theorem:ParityLocalityLower}. The proof is not explicit. Theorem~\ref{Theorem:ParityLocalityUpperD4} gives an explicit family of codes in the narrow case of $d=4.$ We start by introducing some concepts we need for the proof of theorem~\ref{Theorem:ParityLocalityUpperGeneral}.

\begin{definition}
\label{Def:KCore}
Let $L\subseteq \mathbb{F}_q^n$ be a linear space and $S\subseteq [n]$ be a set, $|S|=k.$ We say that $S$ is a $k$-core for $L$ if for all vectors ${\bf v}\in L,$ $\Supp({\bf v})\not\subseteq S.$
\end{definition}

It is not hard to verify that $S$ is a $k$-core for $L,$ if and only if $S$ is a subset of some set of information coordinates in the space $L^{\perp}.$ In other words $S$ is a $k$-core for $L,$ if and only $k$ columns in the $(n-\dim L)$-by-$n$ generator matrix of $L^\perp$ that correspond to elements of $S$ are linearly independent.

\begin{definition}
\label{Def:GenPos}
Let $L\subseteq \mathbb{F}_q^n$ be a linear space. Let $\{{\bf c}_1,\ldots,{\bf c}_n\}$ be a sequence of $n$ vectors in $\mathbb{F}_q^k.$ We say that vectors $\{{\bf c}_i\}$ are in general position subject to $L$ if the following conditions hold:
\begin{enumerate}
\item For all vectors ${\bf v}\in L$ we have $\sum_{i=1}^n{\bf v}(i){\bf c}_i=0;$

\item For all $k$-cores $S$ of $L$ we have $\Rk\left(\{{\bf c}_i\}_{i\in S}\right)=k.$
\end{enumerate}
\end{definition}

The next lemma asserts existence of vectors that are in general position subject to an arbitrary linear space provided the underlying field is large enough.

\begin{lemma}
\label{Lemma:GenPosExist}
Let $L\subseteq \mathbb{F}_q^n$ be a linear space and $k$ be a positive integer. Suppose $q > kn^k;$ then there exists a family of vectors $\{{\bf c}_i\}_{i\in [n]}$ in $\mathbb{F}_q^k$ that are in general position subject to $L.$
\end{lemma}
\begin{proof}
We obtain a matrix $M\in \mathbb{F}_q^{k\times n}$ picking the rows of $M$ at random (uniformly and independently) from the linear space $L^\perp.$ We choose vectors $\{{\bf c}_i\}$ to be the columns of $M.$ Observe that the first condition in definition~\ref{Def:GenPos} is always satisfied. Further observe that our choice of $M$ induces a uniform distribution on every set of $k$ columns of $M$ that form a $k$-core. The second condition in definition~\ref{Def:GenPos} is satisfied as long as all $k$-by-$k$ minors of $M$ that correspond to $k$-cores are invertible. This happens with probability at least
$$
1-\left(n\atop k\right)\cdot \left(1- \prod_{i=1}^k\left(1-\frac{1}{q^i}\right)\right)\geq
1-\left(n\atop k\right)\cdot \left(1- \left(1-\frac{1}{q}\right)^k\right)\geq
1-n^k\cdot \frac{k}{q}>0.
$$
This concludes the proof.
\end{proof}

We proceed to the main result of this section.

\begin{theorem}
\label{Theorem:ParityLocalityUpperGeneral}
Let $2 < d < r+3,$ $r<k,$ $r\mid k.$ Let $q>kn^k$ be a prime power. Let $n=k+\frac{k}{r}+d-2.$ There exists a systematic $[n,k,d]_q$ code $\mc{C}$ of information locality $r,$ where $\frac{k}{r}$ parity symbols have locality $r,$ and $d-2$ other parity symbols have locality $k-\left(\frac{k}{r}-1\right)(d-3).$
\end{theorem}
\begin{proof}
Let $t=\frac{k}{r}.$ Fix some $t+1$ subsets $P_0,P_1,\ldots,P_t$ of $[n]$ subject to the following constraints:
\begin{enumerate}
\item $|P_0|=k-\left(t-1\right)\left(d-3\right)+1;$

\item For all $i\in [t],$ $|P_i|=r+1;$

\item For all $i,j\in [t]$ such that $i\ne j,$ $P_i\cap P_j=\emptyset;$

\item For all $i\in [t],$ $|P_0\cap P_i|=r-d+3.$
\end{enumerate}
For every set $P_i,$ $0\leq i\leq t$ we fix a vector ${\bf v}_i\in\mathbb{F}_q^n,$ such that $\Supp({\bf v}_i)=P_i.$ We ensure that non-zero coordinates of ${\bf v}_0$ contain the same value. We also ensure that for all $i\in [t]$ non-zero coordinates of ${\bf v}_i$ contain distinct values. The lower bound on $q$ implies that these conditions can be met. For a finite set $A$ let $A^\circ$ denote a set that is obtained from $A$ by dropping {\it at most} one element. Note that for all $i\in [t]$ and all non-zero $\alpha,\beta$ in $\mathbb{F}_q$ we have
\begin{equation}
\label{Eqn:DropCombination}
\Supp(\alpha{\bf v}_0+\beta{\bf v}_i)= (P_0\setminus P_i) \sqcup (P_0\cap P_i)^\circ \sqcup (P_i\setminus P_0).
\end{equation}
Consider the space $L=\Span\left(\{{\bf v}_i\}_{0\leq i\leq t}\right).$ Let $M=P_0\setminus \bigsqcup_{i=1}^{t}P_i.$ Observe that
$$|M|=k-(t-1)(d-3)+1-t(r-d+3)=d-2.$$
By~(\ref{Eqn:DropCombination}) for any ${\bf v}\in L$ we have
\begin{equation}
\label{Eqn:LSupport}
\Supp({\bf v})=
\left[
\begin{array}{ll}
\bigsqcup\limits_{i\in T} P_i, & \mbox{ for some $T\subseteq [t]$ OR} \\
M \bigsqcup\limits_{i\in [n]\setminus T} (P_0\cap P_i) \bigsqcup\limits_{i\in T} (P_0\cap P_i)^\circ \bigsqcup\limits_{i\in T} (P_i\setminus P_0) & \mbox{ for some $T\subseteq [t].$}
\end{array}
\right.
\end{equation}
Observe that a set $K\subseteq [n],$ $|K|=k$ is a $k$-core for $L$ if and only if for all $i\in [t],\ P_i\not\subseteq K$ and
\begin{equation}
\label{Eqn:KCoreCriterium}
\left[
\begin{array}{l}
M \not\subseteq K; \mbox{ OR } \\
M \subseteq K\mbox{ and $\exists i\in [t]$ such that}
\left[
\begin{array}{l}
|P_i\cap P_0\cap K|<r-d+2; \mbox{ OR }\\
|P_i\cap P_0\cap K|=r-d+2\mbox{ and } P_i\setminus P_0 \not\subseteq K. \\
\end{array}
\right. \\
\end{array}
\right.
\end{equation}
Let $I\subseteq [n]$ be such that $M\cap I=\emptyset$ and for all $i\in [t],$ $|I\cap P_i|=r.$ By~(\ref{Eqn:KCoreCriterium}) $I$ is a $k$-core for $L.$ We use lemma~\ref{Lemma:GenPosExist} to obtain vectors $\{{\bf c}_i\}_{i\in [n]}\in \mathbb{F}_q^k$ that are in general position subject to the space $L.$ We choose vectors $\{{\bf c}_i\}_{i\in I}$ as our basis for $\mathbb{F}_q^k$ and consider the code $\mathcal{C}$ defined as in~(\ref{Eqn:CodeDefinition}).

In it not hard to see that $\mathcal{C}$ is a systematic code of information locality $r.$ All $t$ parity symbols in the set $\left(\bigsqcup_{i\in [t]}P_i\right)\setminus I$ also have locality $r.$ Furthermore all $d-2$ parity symbols in the set $M$ have locality $k-(t-1)(d-3).$ It remains to prove that the code $\mc{C}$ has distance
\begin{equation}
\label{Eqn:DesiredDist}
d=n-k-t+2.
\end{equation}
According to Fact~\ref{Fact:Dist} the distance of $\mc{C}$ equals $n-|S|$ where $S\subseteq [n]$ is the largest set such that vectors $\{{\bf c}_i\}_{i\in S}$ do not have full rank. By definition~\ref{Def:GenPos} for any $k$-core $K$ of $L$ we have $\Rk\{{\bf c}_i\}_{i\in K}=k.$ Thus in order to establish~(\ref{Eqn:DesiredDist}) it suffices to show that every set $S\subseteq [n]$ of size $k+t-1$ contains a $k$-core of $L.$ Our proof involves case analysis. Let $S\subseteq [n],$ $|S|=k+t-1$ be an arbitrary set. Set
$$b=\#\{i\in [t]\mid P_i\subseteq S\}.$$
Note that since $t(r+1)>|S|$ we have $b\leq t-1.$

{\bf Case 1:} $M\not\subseteq S.$ We drop $t-1$ elements from $S$ to obtain a set $K\subseteq S,$ $|K|=k$ such that for all $i\in [t],$ $P_i\not\subseteq K.$ By~(\ref{Eqn:KCoreCriterium}) $K$ is a $k$-core.

{\bf Case 2:} $M\subseteq S$ and $b\leq t-2.$ We drop $t-1$ elements from $S$ to obtain a set $K\subseteq S,$ $|K|=k$ such that $M\not\subseteq K$ and for all $i\in [t],$ $P_i\not\subseteq K.$ By~(\ref{Eqn:KCoreCriterium}) $K$ is a $k$-core.

{\bf Case 3:} $M\subseteq S$ and $b=t-1.$ Let $i\in [t]$ be such that $P_i\not\subseteq S.$ Such $i$ is unique. Observe that
$$|P_i\cap S| = k+t-1 - (d-2) - (t-1)(r+1)=r-d+2.$$
Also observe that $|P_i\setminus P_0|=r+1-(r-d+3)=d-2\geq 1.$ Combining the last two observations we conclude that either
\begin{equation}
\label{Eqn:CaseThreeUnfold}
\left[
\begin{array}{l}
|P_i\cap P_0\cap S|<r-d+2; \mbox{ OR }\\
|P_i\cap P_0\cap S|=r-d+2\mbox{ and } P_i\setminus P_0 \not\subseteq S. \\
\end{array}
\right. \\
\end{equation}
Finally, we drop $t-1$ elements from $S$ to obtain a set $K\subseteq S,$ $|K|=k$ such that for all $i\in [t],$ $P_i\not\subseteq K.$ By~(\ref{Eqn:KCoreCriterium}) and~(\ref{Eqn:CaseThreeUnfold}) $K$ is a $k$-core.
\end{proof}

Theorem~\ref{Theorem:ParityLocalityUpperGeneral} gave a general construction of $(r,d)$-codes that are optimal not only with respect to information locality and redundancy but also with respect to locality of parity symbols. That theorem however is weak in two respects. Firstly, the construction is not explicit. Secondly, the construction requires a large underlying field. The next theorem gives an explicit construction that works even over small fields in the narrow case of codes of distance $4.$

\begin{theorem}
\label{Theorem:ParityLocalityUpperD4}
Let $r<k,$ $r\mid k$ be positive integers. Let $q\geq r+2$ be a prime power. Let $n=k+\frac{k}{r}+2.$ There exists a systematic $[n,k,4]_q$ code $\mc{C}$ of information locality $r,$ where $\frac{k}{r}$ parity symbols have locality $r,$ and $2$ other parity symbols have locality $k-\frac{k}{r}+1.$
\end{theorem}
\begin{proof}
Fix an arbitrary systematic $[r+3,r,4]_q$ code $\mc{E}.$ For instance, one can choose $\mc{E}$ to be a Reed Solomon code. Let
$$ \mc{E}({\bf y}) = ({\bf y},{\bf p}_0\cdot{\bf y},{\bf p}_1\cdot{\bf y},{\bf p}_2\cdot{\bf y}). $$
Since $\mc{E}$ is a MDS code all vectors $\{{\bf p}_i\}$ have weight $r.$ Thus for some non-zero $\{\alpha_j\}_{j\in [r]}$ we have
\begin{equation}
\label{Eqn:MDSRecover}
{\bf p}_1 = \sum\limits_{j=1}^{r-1}\alpha_j {\bf e}_j + \alpha_r{\bf p}_2,
\end{equation}
where $\{{\bf e}_j\}_{j\in [r]}$ are the $r$-dimensional unit vectors. To define a systematic code $\mc{C}$ we partition the input vector ${\bf x}\in \mathbb{F}_q^n$ into $t=\frac{k}{r}$ vectors ${\bf y}_1,\ldots,{\bf y}_t\in \mathbb{F}_q^r.$ We set
\begin{equation}
\label{Eqn:GluedRS}
\mc{C}({\bf x}) = \left({\bf y}_1,\ldots,{\bf y}_t,{\bf p}_0\cdot {\bf y}_1,\ldots,{\bf p}_0\cdot {\bf y}_t,\left({\bf p}_1\cdot\sum{\bf y}_i\right),\left({\bf p}_2\cdot\sum{\bf y}_i\right)\right),
\end{equation}
where the summation is over all $i\in [t].$ It is not hard to see that the first $k+t$ coordinates of $\mc{C}$ have locality $r.$ We argue that the last two coordinates have locality $k-t+1.$ From~(\ref{Eqn:MDSRecover}) we have
$$
\left({\bf p}_1\cdot\sum{\bf y}_i\right) = \sum\limits_{j=1}^{r-1}\alpha_j \left({\bf e}_j\cdot\sum{\bf y}_i\right) + \alpha_r\left({\bf p}_2\cdot\sum{\bf y}_i\right),
$$
where the summation is over all $i\in [t].$ Equivalently,
$$
\left({\bf p}_1\cdot\sum{\bf y}_i\right) = \sum\limits_{j=1}^{r-1}\sum_{i\in [t]}\alpha_j {\bf y}_i(j) + \alpha_r\left({\bf p}_2\cdot\sum{\bf y}_i\right).
$$
Thus the next-to-last coordinate of $\mc{C}$ can be recovered from accessing $(r-1)t$ information coordinates and the last coordinate. Similarly, the last coordinate can be recovered from $k-t$ information coordinates and the next-to-last coordinate. To prove that the code $\mc{C}$ has distance $4$ we give an algorithm to correct $3$ erasures in $\mc{C}$. The algorithm has two steps.

{\bf Step 1:} For every $i\in [t],$ we refer to a subset $({\bf y}_i,{\bf p}_0\cdot{\bf y}_i)$ of $r+1$ coordinates of $\mc{C}$ as a block. We go over all $t$ blocks. If we encounter a block where one symbol is erased, we recover this symbol from other symbols in the block.

{\bf Step 2:} Observe that after the execution of Step 1 there can be at most one block that has erasures. If no such block exists; then on Step 1 we have successfully recovered all information symbols and thus we are done. Otherwise, let the unique block with erasures be $({\bf y}_j,{\bf p}_0\cdot{\bf y}_j)$ for some $j\in [t].$ Since we know all vectors $\{{\bf y}_i\}_{i \ne j,\ i\in [t]},$ from $\left({\bf p}_1\cdot\sum_{i\in [t]}{\bf y}_i\right)$ and $\left({\bf p}_2\cdot\sum_{i\in [t]}{\bf y}_i\right)$ (if these symbols are not erased) we recover symbols ${\bf p}_1\cdot{\bf y}_j$ and ${\bf p}_2\cdot{\bf y}_j.$ Finally, we invoke the decoding procedure for the code $\mc{E}$ to recover ${\bf y}_j$ form at most $3$ erasures in $\mc{E}({\bf y}_j)=({\bf y}_j,{\bf p}_0\cdot{\bf y}_j,{\bf p}_1\cdot{\bf y}_j,{\bf p}_2\cdot{\bf y}_j).$
\end{proof}

\section{Non-Canonical Codes}\label{Sec:NonCanonicalCodes}

In this section we observe that canonical codes detailed in sections~\ref{Sec:CanonicalCodes} and~\ref{Sec:ParityLocality} are not the only family of optimal $(r,d)$-codes. If one relaxes conditions of theorem~\ref{Theorem:Canonical} one can get other families. One such family that yields uniform locality for all symbols is given below. The (non-explicit) proof resembles the proof of theorem~\ref{Theorem:ParityLocalityUpperGeneral} albeit is much simpler.
\begin{theorem}
\label{Theorem:NonCanonical}
Let $n,k,r,$ and $d\geq 2$ be positive integers. Let $q>kn^k$ be a prime power. Suppose $(r+1)\mid n$ and
$$n-k = \left\lceil \frac{k}{r} \right\rceil +  d-2;$$
then there exists an $[n,k,d]_q$ code where all symbols have locality $r.$
\end{theorem}
\begin{proof}
Let $t=\frac{n}{r+1}.$ We partition the set $[n]$ into $t$ subsets $P_1,\ldots,P_t$ each of size $r+1.$ For every $i\in [t]$ we fix a vector ${\bf v}_i\in\mathbb{F}_q^n,$ such that $\Supp({\bf v}_i)=P_i.$ We set all non-zero coordinates in vectors $\{{\bf v}_i\}_{i\in [t]}$ to be equal to $1.$
We consider the linear space $L=\Span\left(\{{\bf v}_i\}_{i\in [t]}\right).$ For every any ${\bf v}\in L$ we have
$$\Supp({\bf v})=\bigsqcup_{i\in T} P_i\quad \mbox{for some for some $T\subseteq [t].$}$$
Observe that a set $K\subseteq [n],$ $|K|=k$ is a $k$-core for $L$ if and only if for all $i\in [t],\ P_i\not\subseteq K.$ Also observe that conditions of the theorem imply $k\leq n-t.$ Therefore $k$-cores for $L$ exist. We use lemma~\ref{Lemma:GenPosExist} to obtain vectors $\{{\bf c}_i\}_{i\in [n]}\in \mathbb{F}_q^k$ that are in general position subject to the space $L.$ We consider the code $\mathcal{C}$ defined as in~(\ref{Eqn:CodeDefinition}). In it not hard to see that $\mathcal{C}$ has dimension $k$ and locality $r$ for all symbols. It remains to prove that the code $\mc{C}$ has distance
\begin{equation}
\label{Eqn:DesiredDistTwo}
d=n-k- \left\lceil \frac{k}{r} \right\rceil +2.
\end{equation}
Our proof relies on Fact~\ref{Fact:Dist}. Let $S\subseteq [n]$ be an arbitrary subset such that $\Rk\{{\bf c}_i\}_{i\in S}<k.$ Clearly, no $k$-core of $L$ is in $S.$ Let
$$b=\#\{i\in [t]\mid P_i\subseteq S\}.$$
We have $|S|-b\leq k-1$ since dropping $b$ elements from $S$ (one from each $P_i\subseteq S$) turns $|S|$ into an $(|S|-b)$-core. We also have $br\leq k-1$ since dropping one element from each $P_i\subseteq S$ gives us a $br$-core in $S.$ Combining the last two inequalities we conclude that
$$|S|\leq k+\left\lfloor \frac{k-1}{r} \right\rfloor-1.$$
Combining this inequality with the identity $\left\lfloor \frac{k-1}{r} \right\rfloor = \left\lceil \frac{k}{r} \right\rceil-1$ and Fact~\ref{Fact:Dist} we obtain~(\ref{Eqn:DesiredDistTwo}).
\end{proof}

\section{Beyond Worst-Case Distance}\label{Sec:BeyondDist}

\newcommand{\support}{{\text support }}
\newcommand{\bb}{{\ensuremath \mb{b} }}
\newcommand{\cc}{{\ensuremath \mb{c} }}
\newcommand{\ee}{{\ensuremath \mb{e} }}

In this section, codes are assumed to be systematic unless
otherwise stated. They will have $k$ information symbols and $h = n-k$
parity check symbols.

\subsection{Generalized Pyramid Codes}

The supports of the parity check symbols in a code can be described using a bipartite graph. More generally, we define the notion of a set of points with supports matching a graph $G.$

\begin{definition}
Let $G([k],[h],E)$ be a bipartite graph. We say that
$\mb{c}_1,\ldots,\mb{c}_h \in \F_q^k$ have supports matching $G$ if
$\Supp(\cc_j) = \Gamma(j)$ for all $j \in [h]$ where $\Gamma(j)$ denotes the neighborhood of
$j$ in $G$.
\end{definition}

Given points $\mb{c}_1,\ldots,\mb{c}_h$, consider the $k \times h$ matrix $C$ with columns $\mb{c}_1,\ldots,\mb{c}_h$. For $I \subseteq [k]$ and $J \subseteq [h]$, let $C_{I,J}$ denote the sub-matrix of $C$ with rows indexed by $I$ and columns indexed by $J$.

\begin{definition}
Points $\mb{c}_1,\ldots,\mb{c}_h \in \F_q^k$ with supports matching $G$ are in general position if for every $I \subseteq [h]$ and $J \subseteq [k]$ such that there is a perfect matching from $I$ to $J$ in $G$, the sub-matrix $C_{I,J}$ is invertible.
\end{definition}

Standard arguments show that over sufficiently large fields $\F_q$, choosing
$\mb{c}_1,\ldots,\mb{c}_h$ randomly from the set of vectors with
support matching $G$ gives points in general position.

Coming back to codes, the supports of the parity checks define a bipartite graph which we will call the
\support graph. This is closely related to but distinct from the
Tanner graph.

\begin{definition}
Let $\mc{C}$ be a systematic code with point set $C =
\{\mb{e}_1,\ldots,\mb{e}_k,\mb{c}_1,\ldots, \mb{c}_h\}$. The support
graph $G([k],[h],E)$ of $\mc{C}$ is a bipartite graph  where $(i,j) \in E$ if $\mb{e}_i \in \Supp(\mb{c}_j)$.
\end{definition}

For instance in any canonical $(r,d)$-code, the support graph is specified up to relabeling. There are $\frac{k}{r}$ vertices in $V$ of degree $r$ corresponding to $C' \subseteq C$, whose neighborhood partitions the set $U$ and $d-2$ vertices of degree $k$ corresponding to $C'' \subseteq C$. The minimum distance of such a code is exactly $d$, and hence there are some patterns of $d$ erasures that the code cannot correct. However it is possible that the code could correct many patterns of erasures of weight $d$ and higher, for a suitable choice of $\mb{c}_i$s. In general one could ask: among all codes with a support graph $G,$ which codes can correct the most erasure patterns? {\em A Priori}, it is unclear that there should be a single code that is optimal in the sense that it corrects the maximal possible set of patterns. As shown by~\cite{HCL} such codes do exist over sufficiently large fields.

Consider a systematic code $\mc{C}$ with support graph $G([k],[h],E)$. Given $I \subseteq [k]$ and $J
\subseteq [h]$, let $\Gamma_J(I)$ denote $\Gamma(I)\cap J$ (define
$\Gamma_I(J)$ similarly). Consider a set of erasures $S \cup T$ where $S \subseteq [k]$
and $T\subseteq [h]$ are the sets of information and parity check symbols respectively that are erased.
To correct these erasures, we need to recover the symbols
corresponding to $\{{\bf e}_i:i \in S\}$ from the parity checks
corresponding to $\{{\bf c}_j: j \in \bar{T} = [h] \setminus T\}$.
For this to be possible, a necessary condition
is that for every $S^\prime \subseteq S,
|\Gamma_{\bar{T}}(S^\prime)| \geq |S^\prime|$. By Hall's theorem, this
is equivalent to the existance of a matching in $G$ from $S$ to
$\bar{T}$. We say that such a set of erasures satisfies Hall's condition.

\begin{definition}
A systematic code $\mc{C}$ with support graph $G$ is a generalized
pyramid code if every set of erasures satisfying Hall's condition
can be corrected.
\end{definition}

We can rephrase this definition in algebraic terms using the notion of
points with specified supports in general position.

\begin{theorem}
\cite{HCL}
Let $\mc{C}$ be a systematic code with support
graph $G$. $\mc{C}$ is a generalized pyramid code iff
$\mb{c}_1,\ldots,\mb{c}_h$ are in general position with supports
matching $G$.
\end{theorem}

\subsection{The Tradeoff between Locality and Erasure Correction}

For any parity check symbol $\cc_j$, it is clear that $\loc(\cc_j) \leq
\wt(\cc_j) = \deg(j)$. We will show that no better locality
is possible for a generalized pyramid code. This result relies on
a characterization of the support of the vectors in the space $\mc{V}$ spanned by
$\{\cc_1,\ldots,\cc_h\}$ in terms of the graph $G$.

Let $\mc{V}$ denote the space spanned by $\{\cc_1,\ldots,\cc_h\}$ which are in general position with supports matching $G$. Let $\Supp(\mc{V}) \subseteq 2^{[k]}$ denote the set of supports of vectors in $\mc{V}$. We give a necessary condition for membership in $\Supp(\mc{V})$. Our condition is in terms of sets of coordinates that can be eliminated by combination of certain $\cc_j$s.

\begin{definition}
Let $\cc = \sum_{j \in J}\mu_j\cc_j$ where $\mu_j \neq 0$. Let $I = \cup_{j \in J}\Gamma(j) \setminus \Supp(\cc)$. We say that the set $I$ has been eliminated from $\cup_{j \in J} \Gamma(j).$
\end{definition}

\begin{theorem}
\label{thm:support-nec}
Let $\{\cc_1,\ldots,\cc_h\}$ be vectors with supports matching $G$ in general position. The set $I$ can be eliminated from $\cup_{j \in J}\Gamma(j)$ only if $|\Gamma_J(I^\prime)| > |I^\prime|$ for every $I^\prime \subseteq I.$
\end{theorem}
\begin{proof}
Let $\cc = \sum_{j \in J}\mu_j\cc_j$ where $\mu_j \neq 0$. Let $I = \cup_{j \in J}\Gamma(j) \setminus \Supp({\bf c}).$ Assume for contradiction
that there exists $\tilde I \subseteq I$ where $\Gamma_J(\tilde I) \leq |\tilde I|.$ We will show that there exists $I^\prime \subseteq \tilde I$ so that $\Gamma_J(I^\prime) = |I^\prime|$ and that $\Gamma_J(I^{\prime\prime}) > |I^{\prime\prime}|$ for every non-empty subset $I^{\prime\prime} \subsetneq I^\prime.$

It suffices to prove the existence of  $I^\prime \subseteq \tilde I$ where $|\Gamma_J(I^\prime)| = |I^\prime|;$ the claim about subsets of $I^\prime$ will then follow by taking a minimal such $I^\prime.$ Observe that every $i \in I$ must have $|\Gamma_J(i)| \geq 2$, since if $i$ occurs in exactly one $S_j$, then it cannot be eliminated. Hence we must have $|\tilde I| \geq 2.$

One can construct the set $I$ starting from  a singleton and adding elements one at a
time, giving a sequence $I_1,\ldots, I_\ell =I$. We claim
that for any $l \leq \ell$,
$$|\Gamma_J(I_l)| - |I_l| \geq |\Gamma_J(I_{l-1})| - |I_{l-1}| -1.$$
This holds since $\Gamma_J(I_l)$ can only  increase on adding $i$
to $I_{l-1}$ while $|I_l|$ increases by $1$. Since $|\Gamma_J(I_1)| -
|I_1| \geq 1$ whereas $|\Gamma(I_\ell)| - |I_\ell| \leq 0$, we must
have $|\Gamma_J(I_l)| - |I_l| =0$ for some $l \leq \ell$. Thus we have a
set where $|\Gamma_J(I_l) = |I_l|$ as desired.

Since the set $I^\prime$ satisfies Hall's matching condition, there is a perfect matching from $I^\prime$ to $J^\prime = \Gamma_J(I^\prime)$ in $G$. But this means that the sub-matrix $C_{I^\prime,J^\prime}$ has full rank. On the other hand, $\cc = \sum_j\mu_j\cc_j$ and $I^\prime \subseteq I = \cup_j \Gamma(j) \setminus \Supp(\cc)$. Let $\pi(\cc)$  denote the restriction of $\cc$ onto the co-ordinates in $I^\prime$. Then we have
$$\sum_{j \in J} \mu_j\pi(\cc_j) = \sum_{j \in J^\prime}\mu_j\pi(\cc_j) = \pi(\cc) =0.$$
The first equality holds because $\pi(\cc_j) =0$ for $j \not\in \Gamma(I^\prime)$, the second by linearity of $\pi$ and the last since $\pi(\cc) =0$. Hence the vector $\mu_J^\prime= \{\mu_j\}_{j \in J^\prime}$ lies in the kernel of $C_{I^\prime,J^\prime}$ which contradicts the assumption that it has full rank.

This shows that the condition
$|\Gamma_J(I^\prime)| > |I^\prime|$ for all $I^\prime \subseteq I$ is necessary.
\end{proof}

\begin{corollary}
\label{cor:lb}
If the set $I$ can be eliminated from $\cup_{j \in J} \Gamma(j),$ then $|I| \leq |J| -1.$
\end{corollary}

If the field size $q$ is sufficiently large, the necessary condition given by theorem~\ref{thm:support-nec} is also sufficient. We defer the proof of this statement to  Appendix~\ref{Appendix} and  prove our lower bound on the locality of generalized pyramid codes.

\begin{theorem}
\label{thm:gpc}
In a generalized pyramid code, $\loc(\cc_j) = \deg(j)$ for all $j \in [h]$.
\end{theorem}
\begin{proof}
Assume for contradiction that $\loc(\cc_t)  \leq \deg(t) -1$ for some $t \in [h]$. Hence there exist $A \subseteq [k]$ and $B \subseteq [h]$ (not containing $t$) such that
\begin{align*}
\cc_t = \sum_{i \in A}\lambda_i\ee_i  + \sum_{j \in B}\mu_j\cc_j.
\end{align*}
We have $|A| =a$, $|B| = b$ and $a + b \leq \deg(t) -1$.
Hence
\begin{align*}
\cc_t - \sum_{j \in B}\mu_j\cc_j = \sum_{i \in A}\lambda_i\ee_i.
\end{align*}
Thus we have eliminated at least $\deg(t) -a \geq b+1$ indices from $\cup_{j \in B \cup \{t\}}\Supp(\cc_j)$ using a linear combination of $b+1$ vectors. By corollary~\ref{cor:lb} this is not possible for vectors in general position.
\end{proof}

\bibliographystyle{plain}
\bibliography{LDC_PIR_biblio}

\begin{thebibliography}{10}

\bibitem{CJM}
Viveck~R. Cadambe, Syed~A. Jafar, and Hamed Maleki.
\newblock Distributed data storage with minimum storage regenerating codes -
  exact and functional repair are asymptotically equally efficient.
\newblock Arxiv 1004.4299, 2010.

\bibitem{Dimakis_1}
Alexandros~G. Dimakis, Brighten Godfrey, Yunnan Wu, Martin~J. Wainwright, and
  Kannan Ramchandran.
\newblock Network coding for distributed storage systems.
\newblock {\em IEEE Transactions on Information Theory}, 56:4539--4551, 2010.

\bibitem{Dimakis_survey}
Alexandros~G. Dimakis, Kannan Ramchandran, Yunnan Wu, and Changho Suh.
\newblock A survey on network codes for distributed storage.
\newblock {\em Proceedings of the IEEE}, 99:476--489, 2011.

\bibitem{Efremenko}
Klim Efremenko.
\newblock 3-query locally decodable codes of subexponential length.
\newblock In {\em 41st ACM Symposium on Theory of Computing (STOC)}, pages
  39--44, 2009.

\bibitem{HCL}
Cheng Huang, Minghua Chen, and Jin Li.
\newblock Pyramid codes: flexible schemes to trade space for access efficiency
  in reliable data storage systems.
\newblock In {\em 6th IEEE International Symposium on Network Computing and
  Applications (NCA 2007)}, pages 79--86, 2007.

\bibitem{KT}
Jonathan Katz and Luca Trevisan.
\newblock On the efficiency of local decoding procedures for error-correcting
  codes.
\newblock In {\em 32nd ACM Symposium on Theory of Computing (STOC)}, pages
  80--86, 2000.

\bibitem{KSY}
Swastik Kopparty, Shubhangi Saraf, and Sergey Yekhanin.
\newblock High-rate codes with sublinear-time decoding.
\newblock In {\em 43nd ACM Symposium on Theory of Computing (STOC)}, pages
  167--176, 2011.

\bibitem{RSK}
K.~V. Rashmi, Nihar~B. Shah, and P.~Vijay Kumar.
\newblock Optimal exact-regenerating codes for distributed storage at the {MSR}
  and {MBR} points via a product-matrix construction.
\newblock Arxiv 1005.4178, 2010.

\bibitem{TVN_book}
Michael Tsfasman, Serge Vladut, and Dmitry Nogin.
\newblock {\em Algebraic geometric codes: basic notions}.
\newblock American Mathematical Society, Providence, Rhode Island, USA, 2007.

\bibitem{Y_nice}
Sergey Yekhanin.
\newblock Towards $3$-query locally decodable codes of subexponential length.
\newblock {\em Journal of the ACM}, 55:1--16, 2008.

\bibitem{Y_now}
Sergey Yekhanin.
\newblock Locally decodable codes.
\newblock {\em Foundations and trends in theoretical computer science}, 2011.
\newblock To appear. Preliminary version available for download at
  {http://research.microsoft.com/en-us/um/people/yekhanin/Papers/LDC\_now.pdf}.

\end{thebibliography}

\appendix

\section{Spaces spanned by  general position vectors}\label{Appendix}

\begin{lemma}
Let $q \geq n$ be a prime power. The set of supports of vectors in any linear space $\mc{V}\subseteq \mathbb{F}_q^n$ is closed under union.
\end{lemma}
\begin{proof}
Consider two vectors ${\bf a}$ and ${\bf b}$ in $\mathcal{V}$ with $\Supp({\bf a}) =S$ and $\Supp({\bf b})=T.$ We may assume that $|S|, |T| \leq n-1$ and that one set does not contain the other. Now consider ${\bf a} + \lambda {\bf b}$ for $\lambda \in F_q^*.$ It suffices to find $\lambda$ such that ${\bf a}(i) + \lambda {\bf b}(i) \neq 0$ for each $i \in S \cap T.$ This rules out at most $|S \cap T| \leq n-2$ values of $\lambda,$ so there is a solution provided $q -1 > n-2$ or $q \geq n.$
\end{proof}

It is easy to see that the condition $q \geq n$ is tight by considering the length $3$ parity check code over $\F_2$, where the set of supports is not closed under union.

\begin{theorem}
\label{thm:support-suff}
Let $q \geq n$. Let $\{\cc_1,\ldots,\cc_h\}$ be vectors with supports matching $G$ in general position which span a space $\mc{V}$. $\Supp(\mc{V})$ consists of all sets of the form $\cup_{j \in J}\Gamma(j) \setminus I$ where $I$ satisfies the condition $|\Gamma_J(I^\prime)| > |I^\prime|$ for every $I^\prime \subseteq I.$
\end{theorem}
\begin{proof}
Theorem \ref{thm:support-nec} shows that the condition on $I$ is necessary, we now show that it is sufficient. For $j \not\in \Gamma(I)$ the sets $\Gamma(j)$ and $I$ are disjoint. Hence we can write
$$\cup_{j \in J}\Gamma(j) \setminus I = (\cup_{j \in J \cap \Gamma(I)}\Gamma(j)
\setminus I) \bigcup(\cup_{j \in J \cap \overline{\Gamma(I)}}\Gamma(j)).$$
By the closure under union, it suffices to prove the statement in the
case when $J \subseteq \Gamma(I)$. Fix $j_0 \in J$ and let $J^\prime = J\setminus\{j_0\}$.
Since $|\Gamma_J(I^\prime)| > |I^\prime|$ we have $|\Gamma_{J^\prime}(I^\prime)| \geq |I^\prime|$ for every
$I^\prime \subseteq I$. So there is a matching from $I$ to some subset $J^{\prime\prime}
\subseteq J^\prime$ where $|J^{\prime\prime}| = |I|$, and the matrix $C_{I,J^{\prime\prime}}$ is of
full rank since the $\cc_j$s are in general position.

Let $\pi(\cc)$ denote the restriction of a vector $\cc$ onto coordinates in $I$. Since $C_{I,J^{\prime\prime}}$ is invertible, the row vectors $\{\pi(\cc_j)\}_{j \in J^{\prime\prime}}$ have full rank. Note that $\pi(\cc_{j_0})$ is not a zero vector since $j_0 \in \Gamma(I).$ So there exist $\{\mu_j\}$ for $j \in J$ which are not all $0$ and
$$\pi(\cc_{j_0}) = \sum_{j \in J^{\prime\prime}}\mu_j\pi(\cc_j).$$

Now consider the vector $\cc^\prime_{j_0} = \cc_{j_0} - \sum_{j \in J^{\prime\prime}}\mu_j\cc_j$. Note that $\pi(\cc^\prime_{j_0})$ is a zero vector, which shows that $\Supp(\cc^\prime_{j_0}) \subseteq \cup_{j \in J^{\prime\prime} \cup \{j_0\}}\Gamma(j) \setminus I$. We will show that equality holds by using corollary~\ref{cor:lb}. Since we have eliminated $|I|$ vectors, the linear combination must involve at least $|I| +1$ vectors, which means that $\mu_j \neq 0$ for all $j$. Further the set of eliminated co-ordinates cannot be larger than $I$, since this would violate corollary~\ref{cor:lb}. Hence we have
\begin{equation}
\label{eq:supp}
\Supp(\cc^\prime_{j_0}) =  \cup_{j \in J^{\prime\prime} \cup \{j_0\}}\Gamma(j) \setminus I.
\end{equation}

By repeating this argument for every $j_0 \in J$, we will be able to
find $J(j_0) \subseteq J$ of size $|I|+1$ which contains $j_0$ and a
vector $\cc^\prime_{j_0}$ such that
$$\Supp(\cc^\prime_{j_0}) =  \cup_{j \in J(j_0)}\Gamma(j) \setminus I.$$
Using the closure under union of supports, we
conclude that $\Supp(\mc{V})$ contains the set
$$\cup_{j_0 \in J}\cup_{j \in J(j_0)}\Gamma(j) \setminus I = \cup_{j \in J}\Gamma(j) \setminus I.$$
\end{proof}

\end{document}